\documentclass[10 pt, conference]{ieeeconf}

\IEEEoverridecommandlockouts

\usepackage{cite}

\expandafter\let\csname proof\endcsname\relax

\expandafter\let\csname endproof\endcsname\relax

\usepackage{amsmath,amssymb,amsfonts}

\usepackage{graphicx}
\usepackage{flushend}
\usepackage{textcomp}

\usepackage{makecell}

\usepackage{xcolor}

\usepackage{amsthm}

\usepackage{hyperref}

\usepackage{algorithmic}

\usepackage{algorithm}

\usepackage{subfigure}

\usepackage{booktabs}

\usepackage{balance}

\def\BibTeX{{\rm B\kern-.05em{\sc i\kern-.025em b}\kern-.08em

    T\kern-.1667em\lower.7ex\hbox{E}\kern-.125emX}}

\usepackage{threeparttable}

\newtheorem{lemma}{Lemma}
\newtheorem{problem}{Problem}
\newtheorem{remark}{Remark}
\newtheorem{definition}{Definition}

\newtheorem{assumption}{Assumption}

\usepackage{url}

\usepackage{color}

\allowdisplaybreaks

\title{\LARGE \bf
Stable and Safe Reinforcement Learning via a Barrier-Lyapunov Actor-Critic Approach
}

\author{Liqun Zhao, Konstantinos Gatsis, Antonis Papachristodoulou
\thanks{*The {authors} are within the Department of Engineering Science, University of Oxford, Oxford, United Kingdom. E-mails: {\tt\small \{liqun.zhao,\ konstantinos.gatsis,antonis\}@eng.ox.ac.uk}}
}
\begin{document}

\maketitle
\thispagestyle{empty}
\pagestyle{empty}

\begin{abstract}

Reinforcement learning (RL) has demonstrated impressive performance in various areas such as video games and robotics. However, ensuring safety and stability, which are two critical properties from a control perspective, remains a significant challenge when using RL to control real-world systems. In this paper, we first provide definitions of safety and stability for the RL system, and then combine the control barrier function (CBF) and control Lyapunov function (CLF) methods with the actor-critic method in RL to propose a Barrier-Lyapunov Actor-Critic (BLAC) framework which helps maintain the aforementioned safety and stability for the system. In this framework, CBF constraints for safety and CLF constraint for stability are constructed based on the data sampled from the replay buffer, and the augmented Lagrangian method is used to update the parameters of the RL-based controller. Furthermore, an additional backup controller is introduced in case the RL-based controller cannot provide valid control signals when safety and stability constraints cannot be satisfied simultaneously. Simulation results\footnote{The code can be found in the GitHub repository: \url{https://github.com/LiqunZhao/A-Barrier-Lyapunov-Actor-Critic-Reinforcement-Learning-Approach-for-Safe-and-Stable-Control}} show that this framework yields a controller that can help the system approach the desired state and cause fewer violations of safety constraints compared to baseline algorithms.

\end{abstract}

\section{INTRODUCTION}
The remarkable success of reinforcement learning (RL) in solving complex sequential decision-making problems, including Atari games, has inspired researchers to explore its potential in real-world applications, such as robotics~\cite{kober2013reinforcement}. However, the trial-and-error nature of RL may lead agents to exhibit actions resulting in dangerous or harmful consequences during learning. For example, although RL agents are typically trained in simulation platforms, it is crucial to take the challenges that arise in real-world scenarios, such as real-world uncertainties \cite{wang2023trustworthy,wang2023quadue}, into consideration when training RL agents directly in real environments where safety is fundamental. Furthermore, regardless of the controller design method employed, stability is one of the paramount properties that need to be guaranteed given a control system. Stability means the system will stay close or converge to the equilibrium, and it bears a close relationship with the performance of the control system. A lack of stability renders the system ineffective and even dangerous \cite{han2020actor}.

Safety in RL has been gaining attention from researchers. Existing safe RL approaches are mainly based on the constrained Markov decision process (CMDP) \cite{altman1999constrained}, and methods such as primal-dual update \cite{tessler2018reward,paternain2019constrained}, and trust region policy optimization \cite{achiam2017constrained,yang2020projection} have been employed. However, in previous studies, safety constraints are usually defined based on the cumulative cost of an entire trajectory, rather than on the individual cost signals at each timestep along the trajectory. Consequently, safety violations at specific timesteps are wrapped by the trajectory expectation, and certain states may be permitted to be unsafe \cite{ma2021learn,ma2021feasible,ma2022joint}. 

Recently, there has been growing interest in combining control-theoretic methods with RL to ensure the safety for the system \cite{hewing2020learning}. In \cite{koller2018learning}, a learning-based model predictive control method is proposed to provide safety guarantees with high probability. Besides, other concepts such as safe set algorithm (SSA) \cite{ma2021learn,wei2019safe} and control barrier function (CBF) \cite{cheng2019end,emam2021safe,wang2022safety,do2023game} have also been applied as safety constraints to help RL training maintain safety. However, according to \cite{emam2021safe}, the supervised learning of the CBF layer in \cite{cheng2019end} has the potential to introduce approximations that may adversely impact RL training. Also, applying a filter to aid the RL-based controller in choosing a safe action at each timestep can lead to a jerky output, or trajectory, which is often undesirable in many applications. In \cite{emam2021safe}, when the time interval between consecutive timesteps is not small enough, there might be a misalignment between the constrained optimization part, where continuous CBFs are used as constraints, and the discrete-time system dynamics yielding the state transitions for the Markov decision process (MDP). 

The Lyapunov function, a popular tool in the control community, has recently been applied to RL as well \cite{cao2023physical}. Some previous studies develop ways to construct Lyapunov functions for RL training \cite{chow2018lyapunov}, while \cite{berkenkamp2017safe} proposes a method using Lyapunov functions to guide safe exploration, which is however difficult to be applied to high-dimensional systems due to the curse of dimensionality caused by the discretization of the state space. \cite{han2020actor} and \cite{chang2021stabilizing} apply Lyapunov functions to model-free RL to help to guarantee the stability for various systems. However, considering the sample efficiency, and the fact that usually at least a nominal model is available in real applications like robotic arms, model-based RL methods can be used since they are more sample efficient than model-free RL methods.

In this paper, our focus is on helping to guarantee safety and stability for systems that can be expressed by MDP. Our main work and contributions can be summarized as follows:
\begin{itemize}
        \item An RL-based controller is proposed by combining CBF and control Lyapunov function (CLF) with the Soft Actor-Critic (SAC) algorithm \cite{haarnoja2018soft2}, and the augmented Lagrangian method is used to solve the corresponding constrained optimization problem. The RL-based controller assists to guarantee both safety and stability simultaneously with separate CBF constraints and CLF constraint, respectively, and the augmented Lagrangian method can update the controller parameters efficiently while making the hyperparameter tuning for different learning rates easier, which might be difficult for the primal-dual update widely used in previous studies.
        \item A backup controller is proposed to replace the RL-based controller when no feasible solution exists to satisfy both safety and stability constraints simultaneously. The inequality constraints of the backup controller are constructed based on CBFs, which help the system achieve and maintain safety. The CLF constraint is incorporated into the objective function to prevent the system from diverging too far from equilibrium while satisfying CBF constraints. This safety-critical design is well-suited for real-world applications where safety is of paramount importance, and accelerates the system's approach towards its desired state in practice.
        \item A framework called Barrier-Lyapunov Actor-Critic
(BLAC) is proposed by combining the RL-based and backup controllers, and we test it on two simulation tasks. Our results show that the framework helps to guarantee the safety and stability of the system, and therefore achieves better results compared to baselines.
\end{itemize}

While unifying safety and stability is common in many other studies \cite{ames2019control,tan2021undesired, dawson2022learning}, to the best of the authors' knowledge, this is the first time to combine actor-critic RL with CBF and CLF as separate constraints to train an RL-based controller. \cite{choi2020reinforcement} uses RL to learn uncertainties in the CLF, CBF, and other constraints to generate control signals using quadratic programming, and thus is different from the framework proposed in this paper.
 
\section{PROBLEM STATEMENT}
\label{sec:problem_formulation}

The Markov decision process (MDP) with control-affine dynamics can be defined by the tuple $\mathcal{M}$, which is $(\mathcal{X} , \mathcal{U} ,f,g,d , r, c,\gamma, \gamma_{c} )$. $\mathcal{X} \subset \mathbb{R} ^{n} $ and $\mathcal{U} \subset \mathbb{R} ^{m}$ are state and control signal spaces, and the state transitions for the MDP are obtained by the following control-affine system:
\begin{align}\label{affine model}
        x_{t+1}=f(x_t)+g(x_t)u_t+d(x_t).
\end{align}
Here, $x_t \in \mathcal{X} $ is the state at timestep $t$, $u_t \in \mathcal{U} $ is the control signal at timestep $t$, and with an RL-based controller $\pi$, it is sampled from a distribution $\pi(u_t|x_t)$. $f: \mathbb{R} ^{n} \rightarrow \mathbb{R} ^{n}$ and $g: \mathbb{R} ^{n} \rightarrow \mathbb{R} ^{n\times m}$ define the known nominal model of the system. $d: \mathbb{R} ^{n} \rightarrow \mathbb{R} ^{n}$ denotes the unknown model which is continuous with respect to the state. $r$, $c$ are the reward and cost, respectively, and $\gamma$ and $\gamma_{c}$ are the discount factors.
\begin{remark}
    In practical applications, the nominal model that is known beforehand may perform poorly, which necessitates the need to learn a better model. Similar to \cite{cheng2019end,emam2021safe}, we use a Gaussian process (GP), a kernel-based nonparametric regression model, to estimate the unknown dynamics $d$ from data. Consequently, when constructing CBF and CLF constraints in Section \ref{sec:framework design}, we replace $d$ using the mean and variance given by the GP. Due to poor scalability with a large number of data points, we cease GP updates after a certain number of episodes.\end{remark}

Here we present some additional notations that will be used later. Based on (\ref{affine model}), the transition probability can be denoted as $P(x_{t+1}|x_t,u_t)\triangleq I_{\{x_{t+1}=f(x_t)+g(x_t)u_t+d(x_t)\}}$ 
where $I_{\{x_{t+1}=f(x_t)+g(x_t)u_t+d(x_t)\}}$ is an indicator function that equals 1 if $x_{t+1}$ satisfies (\ref{affine model}) given $x_t$ and $u_t$, and 0 otherwise. Similar to \cite{han2020actor}, the closed-loop transition probability is denoted as $P_{\pi}(x_{t+1}|x_t)\triangleq \int_{\mathcal{U} }\pi(u_t|x_t)P(x_{t+1}|x_t,u_t)du_t$. Moreover, the closed-loop state distribution at timestep $t$ is denoted by $\upsilon(x_{t}|\rho,\pi,t)$, which can be calculated iteratively using the closed-loop transition probability: $\upsilon(x_{t+1}|\rho,\pi,t+1)=\int_{\mathcal{X}}P_{\pi}(x_{t+1}|x_t)\upsilon(x_{t}|\rho,\pi,t)dx_{t}, \,\,\,\forall t \in \mathbb{N}$, and $\upsilon(x_{0}|\rho,\pi,0)=\rho$ is the initial state distribution.

\subsection{Definition of Safety}
In this subsection, we present the condition that the controller should satisfy to maintain system safety. Assuming there are $k$ different safety constraints that need to be satisfied, the system is considered safe if
\begin{equation}
    h_i(x_t)  \geq 0 \,\,\,\,\,\,\,\, \forall t \geq 0
    \label{eq:C_definition}
\end{equation}
holds for each $i = 1,\ldots, k$. Here $h_i: \mathbb{R}^n \rightarrow \mathbb{R} $ is a function defined for the $i$-th safety constraint, and a safe set $\mathcal{C}_i \subset \mathbb{R} ^{n}$ can be defined by the super-level set of $h_i$ as follows:
\begin{equation}\label{ith safe set}
        \mathcal{C}_i = \{x \in \mathbb{R}^n|h_i(x) \geq 0\}.
\end{equation}
A safe set $\mathcal{C} \subset \mathbb{R} ^{n}$ can therefore be defined as the intersection of all $\mathcal{C}_i$:
\begin{equation}\label{definition of safe set}
        \mathcal{C}=\bigcap\limits_{i=1}^{k}\mathcal{C}_i=\bigcap\limits_{i=1}^{k}\{x\in\mathbb{R}^n|h_i(x)\ge 0\}.
\end{equation}
We require the system state to remain within this set $\mathcal{C}$, i.e., the safe set $\mathcal{C}$ should be forward invariant. Therefore, the system is required to be safe at every timestep, and thus the constraint here is stricter than that constructed by the expected return of costs, which is widely used in previous studies. A control barrier function can be used to ensure forward invariance of the safe set, and therefore can be naturally applied here to help maintain system safety.
\begin{definition}[Discrete-time Control Barrier Function \cite{cheng2019end}] 
        Given a set $\mathcal{C}_i \subset \mathbb{R} ^{n}$ defined by (\ref{ith safe set}), the function $h_i$ is called a discrete-time control barrier function (CBF) for system (\ref{affine model}) if there exists $\eta \in [0,1]$ such that
        \begin{equation}\label{CBF constraints}
                \mathop {\sup }\limits_{u_t \in \mathcal{U}} \left\{ h_i\big(f(x_t)+g(x_t)u_t+d(x_t)\big) - h_i(x_t)\right\} \ge - \eta h_i(x_t)
        \end{equation}      
        holds for all $x_t\in\mathcal{C}_i$.
\end{definition}
The existence of a  CBF means the existence of a controller such that the set $\mathcal{C}_i$ is forward invariant. Consequently, safety is maintained if there exists a controller such that $\forall i \in [1,k]$, $h_i\big(f(x_t)+g(x_t)u_t+d(x_t)\big) - h_i(x_t) \ge - \eta h_i(x_t)$ holds for all $x_t\in\mathcal{C}$.

\subsection{Definition of Stability} 
In this subsection, we first introduce the cost function, and then give the definition of stability used in this paper.

In a stabilization task,  our goal is to find a controller that can drive the system state to the equilibrium, i.e., the desired state, eventually. 
To achieve this goal, given the state $x_t$ and control signal $u_t$, we define the instantaneous cost signal, which is one element in the tuple  $\mathcal{M}$,  to be $c(x_t,u_t)=\left\lVert x_{t+1}-x_{\text{desired}}\right\rVert$ where $x_{t+1}$ is the next state following~(\ref{affine model}), and $x_{\text{desired}}$ denotes the desired state (i.e., equilibrium). Alternatively, in a tracking task, the control signal can be denoted as $c(x_t,u_t)=\left\lVert x_{t+1}-r_s\right\rVert$  where $r_s$ is the reference signal the system needs to track.

Since we investigate the stability of a closed-loop system under a nondeterministic RL-based controller $\pi$, we combine $\pi(u_t|x_t)$, which is a Gaussian distribution in this paper, with the cost signal $c(x_t,u_t)$ to define the cost function under the controller $\pi$ as 
\begin{align}\label{c_pi definition}
        c_{\pi}(x_t) =  \mathbb{E}_{u_t \sim \pi}c(x_t,u_t)=
        \mathbb{E}_{u_t \sim \pi}[\left\lVert x_{t+1}-x_{\text{desired}}\right\rVert ].
\end{align}
The cost function $c_{\pi}(x_t)$ represents the expected value of the norm of the difference between the next state $x_{t+1}$ following (\ref{affine model}), and the desired state $x_{\text{desired}}$, over $u_t$ sampled from the distribution $\pi(u_t|x_t)$. It is natural to expect that the value of $c_{\pi}(x_t)$ should decrease as $t$ increases to drive the system state towards the equilibrium, and we hope eventually $c_{\pi}(x_t) = 0$, which means the state reaches the equilibrium. However, as the control signal is sampled from a Gaussian distribution, the state at timestep $t$ is also distributed, which necessitates the use of the concept of ``expected value" for $c_{\pi}(x_t)$. Therefore, we adopt the definition of stability as presented in \cite{han2020actor} for the framework proposed in this paper.

\begin{definition}[Stability in Mean Cost]\label{Stability in Mean Cost at the Equilibrium}
Let $\upsilon(x_{t}|\rho,\pi,t)$ denote the closed-loop
state distribution at timestep $t$. The equilibrium of a system  is said to be stable in mean cost if there exists a positive constant $b$ such that for any initial state $x_0 \in \{x_0|c_{\pi}(x_0) < b\}$, the condition
\begin{align}\label{stability definition 1}
 \lim_{t\to\infty}\mathbb{E}_{x_t \sim \upsilon}\big[c_{\pi}(x_t)\big]=0
\end{align}
holds. If $b$ is arbitrarily large,
the equilibrium is globally stable in mean cost.
\end{definition}

\begin{remark}
It is important to note that the term ``stability in mean cost" differs from the commonly used concept of Lyapunov stability. The condition~(\ref{stability definition 1}) constructed by using the expected value is the convergence condition required for asymptotic stability. Therefore, the condition of Lyapunov stability is not involved in Definition~\ref{Stability in Mean Cost at the Equilibrium}. For linear systems, convergence implies stability in the sense of Lyapunov. However, in general cases, this does not hold - see Vinograd's counterexample. Considering this, in Section~\ref{sec:framework design}, we apply an exponentially stabilizing CLF to achieve exponential stability, which implies Lyapunov stability for the equilibrium.
\end{remark}

\subsection{Definition of the Safe and Stable Control Problem}
Based on the previous subsections, similar to \cite{dawson2022safe1}, we give the formal definition of the safe and stable control problem:

\begin{problem}[Safe and Stable Control Problem]\label{Safe and Stable Control Problem}
Given a control-affine system $x_{t+1}=f(x_t)+g(x_t)u_t+d(x_t)$, a unique desired state (equilibrium) $x_{\text{desired}}$, a set $\mathcal{X}_b=\{x_0|c_{\pi}(x_0) < b\}$ where $b$ is an arbitrarily large positive number and $x_0$ denotes the initial state, a set of unsafe states $\mathcal{X}_{unsafe}\subseteq \mathcal{X}$, and a set of safe states $\mathcal{X}_{safe}\subseteq \mathcal{X}$ such that $x_{\text{desired}} \in \mathcal{X}_{safe}$ and $\mathcal{X}_{safe}\cap\mathcal{X}_b \neq \emptyset$, find a controller $\pi$ generating control signal $u_t$ such that all trajectories satisfying $x_{t+1}=f(x_t)+g(x_t)u_t+d(x_t)$ and $x_0 \in \mathcal{X}_{safe}\cap\mathcal{X}_b$ have the following properties:
\begin{itemize}
\item \textbf{Safety}: $x_t \in \mathcal{X}_{safe}\,\,\,\,\forall t \geq 0 $.
\item \textbf{The Equilibrium Is Stable in Mean Cost:} $\lim_{t\to\infty}\mathbb{E}_{x_t \sim \upsilon}\big[c_{\pi}(x_t)\big]=0$ where $\upsilon(x_{t}|\rho,\pi,t)$ is the closed-loop
state distribution at timestep $t$.
\end{itemize}
\end{problem}
Therefore $\mathcal{X}_{safe}$ is a safe set that is forward invariant, and the controller is used to help the system arrive at the desired state $x_{\text{desired}}$ while avoiding the unsafe states.

\begin{remark}
     Since $b$ is arbitrarily large in Problem \ref{Safe and Stable Control Problem}, $x_{\text{desired}}$ is globally stable in mean cost and is the only equilibrium point. Thus, given this equilibrium is stable in mean cost, we can say that the system is stable in mean cost. For brevity, we will use the term ``stability" to refer to ``stability in mean cost" in the rest of this paper, and when we say the system is stable, it means the system is stable in mean cost.
\end{remark}
\begin{remark}
    Given $b$ is arbitrarily large, the problem formulation requires that from every initial state $x_0 \in \mathcal{X}_{safe}$ the system should arrive at the unique equilibrium $x_{\text{desired}}$ eventually. Therefore, this paper excludes tasks where it is impossible for the system to reach the unique equilibrium $x_{\text{desired}}$ starting from some safe states, or modifications should be made to preclude those safe but undesirable states.
\end{remark}

\section{Framework Design}
\label{sec:framework design}
In this section, we first define the value function of the cost and give a condition for stability. Then, we employ the augmented Lagrangian method to update the RL-based controller parameters, aiming to satisfy safety and stability conditions and thus help to guarantee both properties for the system. Finally, considering the possible infeasibility of the constrained optimization problem, a backup controller is proposed to replace the RL-based controller when no feasible solution exists to satisfy both safety and stability constraints simultaneously. The overall Barrier-Lyapunov Actor-Critic (BLAC) framework is included at the end of this section.

\subsection{Value Function of the Cost} 
To assist in maintaining stability for the system, inspired by the commonly-used value functions in RL literature, we define the value function of the cost at the state $x_t$ similarly:
\begin{equation}\label{definiton of Lyapunov function}
        L_{\pi}(x_t)=\mathbb{E}_{\tau \thicksim \pi}\big[\sum_{i=0}^{\infty}\gamma_c^{i}c_{\pi}(x_{t+i})\big].  
    \end{equation}
Here $\tau \!=\!\{x_{t},x_{t+1},x_{t+2},\cdots\}$ is a trajectory under controller $\pi$ starting from the initial state $x_t$. Based on this definition, $L_{\pi}(x_t)$ can also be approximated by a neural
network, and since $c_{\pi}(x_{t})$ is non-negative for any $t \geq 0$, we may choose to apply the rectified linear unit (ReLU) as the output activation function to make $L_{\pi}(x_t) \geq 0$ for each $x_t$. Other methods to ensure non-negativity can be found in \cite{han2020actor}, \cite{dawson2022safe2}. 

Based on (\ref{definiton of Lyapunov function}), to achieve stability, a natural approach is to consider imposing a condition in the algorithm that ensures that the value of $L_{\pi}(x_t)$ decreases along the trajectory $\tau$. Inspired by the concept of exponentially stabilizing CLF and \cite{han2020actor}, we first make two assumptions for the MDP:
\begin{assumption}\label{Finite Reward and Cost}
       The state and control signal are sampled from compact sets, and the reward and cost values obtained at any timestep are bounded by $r_{max}$ and $c_{max}$, respectively. Therefore, the value function of the reward and cost are upper bounded by $\frac{r_{max}}{1-\gamma}$ and $\frac{c_{max}}{1-\gamma_c}$, respectively.
\end{assumption}
\begin{assumption}[Ergodicity]\label{ergodicity}
        The Markov chain induced by controller $\pi$ is ergodic with a unique stationary distribution $q_{\pi}(x) = \lim_{t\to\infty}\upsilon(x_t=x|\rho,\pi,t)$, where $\upsilon(x_t|\rho,\pi,t)$ is the closed-loop state distribution.
\end{assumption}
These assumptions are widely used in previous RL research. For more content on the relationship between models of control systems and Markov chains, as well as ergodicity, we refer the interested reader to the book \cite{meyn2022control}. Then, we introduce the Lemma \ref{Lemma 1} invoked by \cite{han2020actor} as follows:
\begin{lemma} \label{Lemma 1}
Under Assumptions \ref{Finite Reward and Cost}, \ref{ergodicity}, the system defined in (\ref{affine model}) is stable in mean cost if there exist positive constants $\alpha_1$, $\alpha_2$, $\beta$, and a controller $\pi$, such that
\begin{align}\label{stability requirement 1}
            \alpha_1 c_{\pi}(x) \leq L_{\pi}(x) \leq \alpha_2 c_{\pi}(x)
        \end{align}
        \begin{align}\label{stability requirement 2}
                \mathbb{E}_{x\sim \mu_{\pi},x^\prime\sim P_{\pi}}\big[L_{\pi}(x^\prime) \! - \! L_{\pi}(x)\big]
                \!  \leq \! -\beta\mathbb{E}_{x\sim \mu_{\pi}}\big[L_{\pi}(x)\big]
        \end{align}
        hold for all $x \in \mathcal{X} $. Here $L_{\pi}$ defined in (\ref{definiton of Lyapunov function}) is the value function of the cost under the controller $\pi$, and 
        \begin{align}\label{definiton of sampling distribution}
                \mu_{\pi}(x)  \triangleq \lim_{N\to\infty}\frac{1}{N}\sum_{t=0}^{N}\upsilon(x_t=x|\rho,\pi,t)
        \end{align}
        is the sampling distribution, where $\upsilon(x_{t}|\rho,\pi,t)$ is the closed-loop
state distribution at timestep $t$.
\end{lemma}
\begin{proof}
The proof of this lemma closely resembles that of Theorem 1 in \cite{han2020actor}, and interested readers are encouraged to refer to that paper for further details.   
\end{proof}
According to the proof presented in \cite{wang2023rl}, $L_{\pi}$ naturally satisfies the constraints (\ref{stability requirement 1}) under Assumption \ref{Finite Reward and Cost}. In the next subsection, we present a method to obtain a controller that results in an $L_{\pi}$ satisfying (\ref{stability requirement 2}), and thus we can call $L_{\pi}$ an exponentially stabilizing CLF, abbreviated as CLF hereafter.

\subsection{Augmented Lagrangian Method for Parameter Updating}
In the previous sections, we proposed conditions (\ref{CBF constraints}) and (\ref{stability requirement 2}) that a controller should satisfy to help guarantee safety and stability. We now turn our attention to updating an RL-based controller $\pi$ to meet these conditions, namely finding a qualified controller. This can be thought of as a constrained optimization problem with the above conditions as constraints. It is common to employ the primal-dual method to update the RL-based controller parameters in previous studies, however, this method often requires hyperparameter tuning, which might be challenging, for adjusting the learning rates used to update the parameters of different neural networks and Lagrangian multipliers. Here we adopt the augmented Lagrangian method, inspired by its effectiveness in solving constrained optimization problems and \cite{9683088}, to update the parameters of the RL-based controller.

Practically, we need to construct the conditions based on data collected during the learning process. At each timestep, the system applies a control signal and receives feedback, including reward and cost. The resulting transition pair $(x_t,u_t,r_t,c_t,x_{t+1})$ is stored in the replay buffer, and we sample a batch of these transition pairs, denoted as $\mathcal{D}$, randomly from this replay buffer at each timestep to construct the CBF and CLF constraints with the system model for safety and stability, respectively: 
\begin{equation}\label{inequality constraints for CBF and CLF}
        \begin{array}{l}
            h_i(\hat{x}_{t+1}) - h_i(x_t) \ge -\eta h_i(x_t) \vspace{1ex}
      \qquad \qquad \qquad \,\,\, \forall i \in [1,k],\\\vspace{1ex}
         L_{\pi}(\hat{x}_{t+1})-L_{\pi}(x_t) \leq -\beta L_{\pi}(x_t),
        \end{array}
        \end{equation}
for each $x_t \in \mathcal{D}$. $\hat{x}_{t+1}=f(x_t)+g(x_t)u+d(x_t)$ is the predicted next state, and $d(x_t)$ is estimated and replaced by the mean value given by GP. Thus, the system model is used in RL training, and these constraints are functions of the controller $\pi$. Note that $u$ here is the control signal generated by the current controller instead of being the historical control signal $u_t$ generated by a previous controller and then stored in a transition pair in $\mathcal{D}$. To apply the augmented Lagrangian method, these inequality constraints are converted to equality constraints using ReLU\cite{9683088} for each $x_t \in \mathcal{D}$: 
        \begin{equation}\label{equality constraints for CBF and CLF with relu}
                \begin{array}{l}
                        ReLU\big(h_i(x_{t}) - h_i(\hat{x}_{t+1}) -\eta h_i(x_t)\big) = 0 \vspace{1ex} 
        \qquad   \forall i \in [1,k],\\\vspace{1ex}
                    ReLU\big(L_{\pi}(\hat{x}_{t+1}) - L_{\pi}(x_t) +\beta L_{\pi}(x_t)\big)= 0.
                \end{array}
                \end{equation}
        Then the actor-critic approach is used to learn the RL-based controller. We represent the parameters of the RL-based controller and two action-value networks used in SAC by $\theta$ and $\phi_i, i=1,2$, respectively.  Moreover, we use $L_{\nu}$, which is called Lyapunov network, to approximate $L_{\pi}$ defined in (\ref{definiton of Lyapunov function}) with parameters $\nu$. Using these notations, we formulate a new constrained optimization problem as follows:
\begin{equation}\label{RL-based constrained optimization}
        \begin{split}
            \min_{\theta} \,\, & -V^{\pi_{\theta}}\\
            s.t. 
            \,\,& \mathbb{E}_{x_t\sim \mathcal{D}, u \sim \pi} \! \big[ReLU  \! \big(h_i(x_{t}) \! -  \!h_i(\hat{x}_{t+1})  \! -  \! \eta h_i(x_t)\big)\big]  \! =  \! 0\vspace{1ex} \\\vspace{1ex}
        &\qquad \qquad \qquad \qquad \qquad \qquad \qquad \qquad \qquad  \forall i \in [1,k]\\\vspace{1ex}
            \,\,& \mathbb{E}_{x_t\sim \mathcal{D}, u \sim \pi} \! \big[ReLU  \! \big(L_{\nu}(\hat{x}_{t+1})  \! -  \! L_{\nu}(x_t)  \! +  \! \beta L_{\nu}(x_t)\big)\big]  \! =  \! 0. 
        \end{split}
    \end{equation}
Here expected values are calculated to construct constraints since we sample a batch of $x_t$ from the replay buffer. The objective function $V^{\pi_{\theta}}(x_t)$ is:
        \begin{equation}\label{RL-based controller objective function}
                \begin{split}
                        V^{\pi_{\theta}}=\mathbb{E}_{x_t\sim \mathcal{D},\xi \sim \mathcal{N}}&\big[\mathop {\min }\limits_{j=1,2}Q_{\phi_j}(x_t,\tilde{u}_{\theta}(x_t,\xi ))\\
                        &-\alpha \log \pi_{\theta}(\tilde{u}_{\theta}(x_t,\xi )|x_t)  \big]
                \end{split}
            \end{equation}
which is the same as the commonly-used objective function in SAC \cite{haarnoja2018soft2} with only small differences in notation. $\tilde{u}_{\theta}(x_t,\xi )=\tanh (\mu_{\theta}(x_t)+\sigma_{\theta}(x_t)\odot \xi),\xi \sim \mathcal{N}(0,I)$, where $\mu_{\theta}$ and $\sigma_{\theta}$ denote the mean and standard deviation of the controller $\pi$, which is a Gaussian distribution, and $\odot$ represents element-wise multiplication. Additionally, loss functions of the action-value networks $Q_{\phi_{i}}, i=1,2$, coefficient $\alpha$, and Lyapunov network $L_{\nu}$ are:
        \begin{align}
                \!\!\!\!&\!J_{Q}(Q_{\phi_i})\!\!=\!\mathbb{E}_{(x_t, u_t, r_t, x_{t+1})\sim \mathcal{D},\xi \sim \mathcal{N}}\!\Big[\!\big[r_t \!\!+\! \gamma \big(\!\mathop {\min }\limits_{j=1,2}\!Q_{targ,\phi_j}\!(x_{t+1}, \nonumber\\
                \!\!\!\!&\tilde{u}_{\theta}(x_{t+1},\xi ))\!-\!\alpha \log \pi_{\theta}(\tilde{u}_{\theta}(x_{t+1},\xi )|x_{t+1})\big)\!-\!Q_{\phi_{i}}(x_t,u_t)\big]^2\!\Big], \label{RL-based controller value function loss}\\
                \!\!\!\!&\!J_{\alpha}(\alpha)\!\!=\! - \alpha \times \mathbb{E}_{x_t\sim \mathcal{D},\xi \sim \mathcal{N}}\!\big[\log \pi_{\theta}(\tilde{u}_{\theta}(x_t,\xi )|x_t) \!+\!\mathcal{H}\big], \label{RL-based controller alpha function loss}\\
                \!\!\!\!&\!J_{L}(L_{\nu})\!\!=\!\mathbb{E}_{(x_t, c_t, x_{t+1})\sim \mathcal{D}}\!\Big[\!\big[c_t \!+\! \gamma_c L_{targ,\nu}(x_{t+1}) 
                \!-\!L_{\nu}(x_t)\big]^2\!\Big],\label{Lyapunov  function loss}
        \end{align}
where $Q_{targ,\phi_i}, i=1,2$ are the target action-value networks, and $\mathcal{H}$ is a designed threshold set to be the lower bound of the entropy of the controller $\pi_{\theta}$. The constrained optimization problem (\ref{RL-based constrained optimization}) can be solved by the augmented Lagrangian method, and the augmented Lagrangian function is:
        \begin{equation}\label{augmented Lagrangian function}
                \begin{split}
                    &\mathcal{L}_A (\theta,\lambda_{i},\zeta ;\rho_{\lambda_i},\rho_{\zeta }) = -V^{\pi_{\theta}}\\
                    &+\!\sum_{i = 1}^{k}\! \lambda_{i}\! \times\! \mathbb{E}_{x_t\sim \mathcal{D}, u \sim \pi}\big[ReLU\!\big(h_i(x_{t})\! -\! h_i(\hat{x}_{t+1}) \!-\!\eta h_i(x_t)\big)\big]\\
                    &+ \!\sum_{i = 1}^{k}\!\frac{\rho_{\lambda_i}}{2}\! \big[\mathbb{E}_{x_t\sim \mathcal{D}, u \sim \pi}\!\big[\!ReLU\!\big(h_i(x_{t})\!\! - \!\!h_i(\hat{x}_{t+1})\! -\!\eta h_i(x_t)\big)\big]\big]\!^2\\
                    &+\!\zeta\! \!\times \!\mathbb{E}_{x_t\sim \mathcal{D}, u \sim \pi}\big[ReLU\!\big(L_{\nu}(\hat{x}_{t+1})\! -\! L_{\nu}(x_t) \!+\!\beta L_{\nu}(x_t)\big)\big] \\
                    &+\! \frac{\rho_{\zeta }}{2}\! \big[\mathbb{E}_{x_t\sim \mathcal{D}, u \sim \pi}\big[ReLU\!\big(L_{\nu}(\hat{x}_{t+1}) \!-\! L_{\nu}(x_t) \!+\!\beta L_{\nu}(x_t)\big)\big]\big]\!^2,
                \end{split}
            \end{equation}
where $\lambda_i$ and $\zeta $ are the Lagrangian multipliers for CBF and CLF constraints, respectively, and $\rho_{\lambda_i}$ and $\rho_{\zeta}$ are the corresponding coefficients for the additional quadratic terms. Since it is usually not easy to solve the problem
        \begin{align}\label{argmin update for theta}
                \theta_{k+1} = \mathop{\arg\min}\limits_{\theta} \mathcal{L}_A (\theta,\lambda_{i,k},\zeta_k ;\rho_{\lambda_{i,k}},\rho_{\zeta_{k}}) 
            \end{align}
            directly in RL, we still apply gradient descent to update $\theta$, and gradient ascent to update $\lambda_{i}$ and $\zeta$, while increasing the value of $\rho_{\lambda_i}$ and $\rho_{\zeta }$ gradually to find a solution for the constrained optimization problem (\ref{RL-based constrained optimization}). The pseudocode is
provided as a part of Algorithm \ref{alg:1}.

\subsection{Backup Controller Design}
Due to the existence of several constraints, the feasibility of the constrained optimization problem (\ref{RL-based constrained optimization}) becomes a crucial problem during the learning process. Typically, there are two scenarios where infeasibility can cause problems: 
\begin{itemize}
        \item The CLF constraint for stability is violated, for example, the agent is required to cross an obstacle to reach the equilibrium by the stability constraint, but this is prevented by the safety constraint, and thus the agent is trapped in some specific positions near the obstacle.
        \item The CBF constraint for safety is violated, for example, the desired state is not fixed but a reference signal and falls in a danger region, which means satisfying the CLF constraint will breach the safety constraint eventually.
\end{itemize} 
The frequent invalid control signals provided by the RL-based controller due to the infeasibility may prevent the system from approaching its desired state quickly, or violate the safety constraints severely. Given that safety takes priority when safety and stability constraints cannot be satisfied simultaneously,  similar to \cite{cheng2019end,emam2021safe}, we propose to design a backup controller by formulating an additional constrained optimization problem that leverages CBFs as constraints to achieve and maintain safety. However, compared to previous studies, the objective function of this backup controller is designed not only to minimize the difference between the actual control signal and nominal control signal, but also to prevent the system from deviating too much from its equilibrium by incorporating the CLF constraint. We first establish a QP-based controller as follows:
\begin{equation}\label{formulation of QP-based backup controller}
        \begin{split}
            \min_{u_{\text{modi}},\epsilon } \,  &\frac{1}{2}\,u_{\text{modi}}^T\,Q\, u_{\text{modi}} + k_{\epsilon,i }\epsilon_i ^2 \!-\!  \kappa \nabla_{x}  L_{\nu}(x_t) 
\! \cdot \! g(x_t)u_{\text{modi}}\\
            s.t. 
            & \,\,h_i\big(f(x_t)+g(x_t)(u_{\text{nominal}} - u_{\text{modi}})+d(x_t)\big)\\
            & \,\,\,\,\,\,\,\,\,\,- h_i(x_t)\geq - \eta h_i(x_t)-\epsilon_i \,\,\,\,\, \forall i \in [1,k],
        \end{split}
    \end{equation}
where $x_t$ is the current state of the system, $Q$ is a symmetric positive semidefinite matrix, and similar to \cite{cheng2019end,emam2021safe}, $d(x_t)$ is estimated and replaced by the mean and variance given by GP. $\epsilon_i$ is the slack variable introduced to enforce the feasibility of the constrained optimization problem with the coefficient $k_{\epsilon,i }$. $u_{\text{modi}} = u_{\text{nominal}} - u_{\text{actual}}$, where $u_{\text{modi}}$ is the optimization variable, $u_{\text{nominal}}$ is the nominal control signal, and $u_{\text{actual}}$ is the actual control signal to perform. Additionally, $\nabla_{x}   L_{\nu}(x)$ is the gradient of the Lyapunov network with respect to the state, and therefore the CLF constraint is incorporated into the objective function with the coefficient $\kappa$. 

The choices of coefficients, nominal control signal, and the condition for using the backup controller depend on the systems to which the backup controller is applied. This QP-based controller is tested and shown to work well in tasks where the CBFs are affine with respect to the control signal. When the CBFs are not affine, it is possible to apply local first-order linearization to the CBFs to obtain approximate affine surrogate CBFs, or other types of constrained optimization problems can be utilized to construct the backup controller, for example, using a quadratically constrained quadratic program (QCQP) when the CBF takes a quadratic form with respect to the control signal. Two detailed examples corresponding to the two typical scenarios listed earlier in this subsection are included in Section \ref{sec:experiment}. 

In summary, the framework combining the RL-based and backup controllers can be summarized as Algorithm \ref{alg:1}.
\begin{algorithm}
    
        \caption{Barrier-Lyapunov Actor-Critic (BLAC)} \label{alg:1}
        \begin{algorithmic}[1]
            \STATE Initialization: RL-based controller network $\pi_{\theta} $, coefficient $\alpha$, action-value networks $Q_{\phi_i}, i=1,2 $, Lyapunov network $L_{\nu }$, Lagrange multipliers $\lambda_{i}$ and $\zeta $, replay buffer $\mathcal{B}$,  coefficients of quadratic terms $\rho_{\lambda_i}$ and $\rho_{\zeta }$, learning rates $\eta_{1}$, $\eta_{2}$, and $\eta_{3}$, quadratic term
coefficient factor  $C_{\rho} \in (1,\infty )$
            \FOR{$k = 1,\ldots , K $}
            \IF{Backup controller should be used according to the condition specific to the task}
            \STATE Solve the constrained optimization problem (\ref{formulation of QP-based backup controller})
            \STATE Apply the control signal $u_{\text{actual}}$
            \ELSE 
            \STATE Sample and apply control signal $u_t$
            \STATE Store the transition pair $(x_t,u_t,r_t,c_t,x_{t+1})$ in $\mathcal{B}$
            \STATE Sample a batch of transition pairs randomly from $\mathcal{B}$, and construct CBF and CLF constraints with the system model
            \STATE Update the Lyapunov network and action-value networks by using (\ref{Lyapunov  function loss}) and (\ref{RL-based controller value function loss}) according to
            \vspace{-0.5em}
            \begin{align*}
                 &\nu_{k+1}\leftarrow \nu_{k} - \eta_{1}\nabla_{\nu}J_{L}(L_{\nu_{k}})\\
                & \phi_{i_{k+1}}\leftarrow \phi_{i_{k}} - \eta_{1}\nabla_{\phi_{i}}J_{Q}(Q_{\phi_{i_{k}}})
        \end{align*}
        \vspace{-1.5em}
            \STATE Update the controller network and coefficient $\alpha$ by using (\ref{augmented Lagrangian function}) and (\ref{RL-based controller alpha function loss}) according to
            \vspace{-0.5em}
            \begin{align*}
                 &\theta_{k+1}\leftarrow \theta_{k} - \eta_{2}\nabla_{\theta}\mathcal{L}_A (\theta_k,\lambda_{i,k},\zeta_k ;\rho_{\lambda_{i,k}},\rho_{\zeta_{k} })\\
                & \alpha_{k+1}\leftarrow \alpha_{k} - \eta_{2}\nabla_{\alpha}J_{\alpha}(\alpha_k)
        \end{align*}
        \vspace{-1.5em}
            \STATE  Update the Lagrangian multipliers using (\ref{augmented Lagrangian function}) according to
            \vspace{-0.5em}
            \begin{align*}
                 &\lambda_{i,k+1}\leftarrow \lambda_{i,k} + \eta_{3}\nabla_{\lambda_{i}}\mathcal{L}_A (\theta_{k+1},\lambda_{i,k},\zeta_k ;\rho_{\lambda_{i,k}},\rho_{\zeta_{k} })\\
                & \zeta_{k+1}\leftarrow \zeta_{k} + \eta_{3} \nabla_{\zeta}\mathcal{L}_A (\theta_{k+1},\lambda_{i,k},\zeta_k ;\rho_{\lambda_{i,k}},\rho_{\zeta_{k} })
        \end{align*}
        \vspace{-1.5em}
            \STATE Update coefficients of quadratic terms by $\rho_{\lambda_{i,k+1}}\leftarrow C_{\rho}\rho_{\lambda_{i,k}}$, $\rho_{\zeta_{k+1}}\leftarrow C_{\rho}\rho_{\zeta_{k}}$, and GP model
            \ENDIF
            \ENDFOR 
            \RETURN $\pi_{\theta}$, $Q_{\phi_i}, i=1,2$, and $L_{\nu}$.  
        \end{algorithmic} 
    \end{algorithm}

\section{Simulations}
\label{sec:experiment}
In this section, we test the framework on two tasks to answer the following questions:
\begin{itemize}
        \item Does the BLAC framework assist in guaranteeing the stability for the system when compared to other baseline algorithms? Since in these tasks high rewards will be given when the system approaches and achieves the desired state (equilibrium), we use the cumulative reward as a measure, where a higher cumulative reward indicates that the system can converge to the equilibrium better.
        \item Does the BLAC framework cause fewer violations of safety constraints compared to baseline algorithms?
\end{itemize}

We use LAC \cite{han2020actor}, CPO \cite{achiam2017constrained}, PPO-Lagrangian and TRPO-Lagrangian \cite{Ray2019} as baselines. These baselines use the Lyapunov-based method, trust region policy optimization, and primal-dual method, respectively, and are representative works in the literature of constrained reinforcement learning. Furthermore, to demonstrate whether the CLF constraint contributes to maintaining stability and therefore improves performance, we remove the CLF constraint in the BLAC framework to create an additional algorithm Barrier Actor-Critic (BAC) as another baseline. Two custom environments built based on \cite{cheng2019end}\cite{emam2021safe} are used since model knowledge is required for the proposed method.

\subsection{Unicycle}
This experiment is modified from the first environment in \cite{emam2021safe}. In this experiment, a unicycle is required to arrive at the desired location, i.e., destination, while avoiding collisions with obstacles. The model of the unicycle is given as:
\begin{equation*}
        x_{t+1}=x_t+
        \left[ \begin{array}{cc}
        \Delta T\cos (\theta_t)  & 0 \\
        \Delta T\sin (\theta_t) & 0 \\
        0 & \Delta T 
        \end{array}
        \right ] (u_t+u_{d,t}).
        \end{equation*}
        \begin{figure}
                \centering
                \includegraphics[scale=0.35]{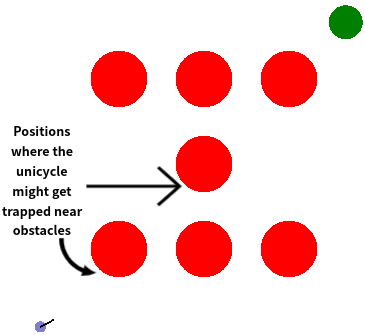}
                \caption{The unicycle environment. The blue point is the unicycle system, the green circle is the destination (desired state), and the red circles are obstacles to avoid. Some positions where the unicycle can get trapped due to the infeasibility of the constrained optimization problem with multiple CBF and CLF constraints are shown.}
                \label{fig:The unicycle environment}
\vspace{-0.5cm}        
\end{figure}

Here $x_t = [x_{1t},x_{2t}, \theta_{t}]^T$ where $x_{1t}$ and $x_{2t}$ are the X-coordinate and Y-coordinate of the unicycle at the timestep $t$, $\theta$ is the angle between the X-coordinate and the direction of the unicycle's movement at $t$. $u_t = [v_t,\omega_t ]^T$ is the control signal where $v_t$ and $\omega_t$ are the linear and angular velocities, respectively. $\Delta T$ represents the time interval. $u_{d,t}=-0.1[\cos(\theta_t),0]^T$ is unknown to the nominal model, and therefore is the unknown part and GPs can be used. Then, similarly to \cite{emam2021safe}, to formulate collision-free safety constraints, a point at a distance $l_p \geq 0$ ahead of the unicycle is considered, and we define the function $p:\mathbb{R}^3 \rightarrow \mathbb{R}^2$ to be
\begin{equation*}
        p(x_t)=
        \left[ \begin{array}{c}
                x_{1t} \\
                x_{2t} \\
        \end{array} 
        \right ]+l_p
        \left[ \begin{array}{c}
                \cos (\theta_t) \\
                \sin (\theta_t) \\
        \end{array}
        \right ].
        \end{equation*}
        
The reward signal is defined as $-K_1(v_t-v_{s})^2 + K_2\Delta d$, where $v_{s}$ is the predefined velocity, $\Delta d$ is the decrease in the distance between the unicycle and destination in two consecutive timesteps, and $K_1$ and $K_2$ are coefficients set to 0.1 and 30, respectively. The cost signal is $\left\lVert p(x_{t+1})-p(x_{\text{desired}})\right\rVert $ where $p(x_{\text{desired}})=[x_{\text{1desired}},x_{\text{2desired}}]^T$ denotes the position of the desired location. CBFs are defined as $h_i(x_t) =\frac{1}{2}\big((p(x_t)-p_{\text{obs}_i})^2 - \delta ^2\big)$ where $p_{\text{obs}_i}$ is the position of the $i$-th obstacle, and $\delta$ denotes the minimum required distance between the unicycle and obstacles. When stability constraint is violated if safety and stability constraints cannot be satisfied simultaneously, the unicycle can get trapped near obstacles, and then the RL-based controller is replaced by the backup controller where $u_{\text{nominal}}$ is set to be the maximum allowable control signal to encourage exploration while maintaining safety. The RL-based controller will resume when the unicycle moves away from the trapped position for a long distance, or when the predetermined time threshold for using the backup controller is exceeded. 

\begin{figure*}[ht]
        \centering
        \includegraphics[scale=0.403]{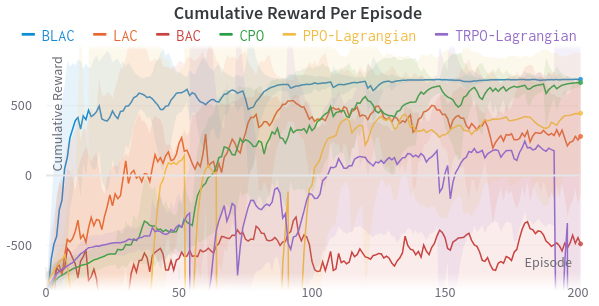}
        \hspace{5mm}
        \includegraphics[scale=0.403]{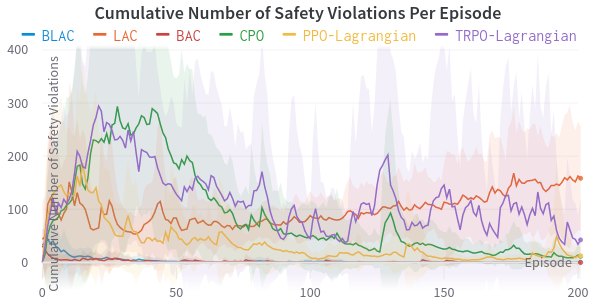}
        \caption{The cumulative reward and cumulative number of safety violations of each episode
in the unicycle environment are compared for the proposed BLAC (in blue) and other baselines. Each plot shows the mean of ten experiments using different seeds, with the shading representing the standard deviation.}
\label{fig:Unicycle results}
\vspace{-0.4cm} 
      \end{figure*}
Simulation results are shown in Figure \ref{fig:Unicycle results}. Compared to other baselines, our framework enables the system to achieve higher cumulative reward in fewer episodes, indicating that the framework helps the unicycle approach and successfully reach its destination (equilibrium) after a shorter training process. Additionally, the fluctuations of the cumulative reward are smaller than those of any other baseline algorithm, suggesting that the system performance in obtaining high rewards can quickly recover to its original high level after deterioration. Regarding safety, as evidenced by the cumulative number of safety violations per episode, our framework results in much fewer violations compared to LAC, CPO, PPO-Lagrangian, and TRPO-Lagrangian where CBFs are not used, which means CBFs can greatly help maintain safety, and our framework is thus suitable for safety-critical applications in the real world. 

\subsection{Simulated Car Following}
This environment, which involves a chain of five cars following each other on a straight road, is adapted from \cite{cheng2019end}\cite{emam2021safe}. The goal is to control the velocity of the $4^{th}$ car to keep a desired distance from the $3^{rd}$ car while avoiding collisions with other cars. The model of cars except for the $4^{th}$ one is:
        \begin{equation*}
                \begin{array}{l}
                        x_{t+1,i}=x_{t,i}
        +
        \left[ \begin{array}{c}
        v_{t,i} \\
        0
        \end{array}
        \right ] \Delta T +
        \left[ \begin{array}{c}
            0  \\
            1+d_i 
            \end{array}
            \right ] a_{t,i} \Delta T \\
        \qquad \qquad \qquad \qquad \qquad \qquad \qquad \qquad  \forall i \in \{1,2,3,5\}.
                \end{array}
                \end{equation*}
Each state of the system is denoted by $x_{t,i} = [p_{t,i},v_{t,i}]^T$, where $p_{t,i}$ and $v_{t,i}$ are the position and velocity of the $i^{th}$ car at the timestep $t$, $d_i=0.1$ is unknown to the nominal model and therefore is the unknown part. $\Delta T$ represents the time interval. The $1^{st}$ car has a velocity $v_{t,1} = v_{s} - 4\sin (t)$, where $v_{s}=3.0$ is the predefined velocity. Its acceleration is given by $a_{t,1}=k_v(v_{s}-v_{t,1})$ where $k_v=4.0$ is a constant. Car 2 and 3 have accelerations given by:
    \begin{equation*}
    a_{t,i}\!=\!\left\{
\begin{aligned}  
&\!k_v(v_{s}\!-\!v_{t,i})\!-\!k_b(p_{t,i\!-\!1}\!-\!p_{t,i})\,\,if\,|p_{t,i\!-\!1}\!-\!p_{t,i}|\! <\! 6.5\\
&\!k_v(v_{s}\!-\!v_{t,i})\,\,\,\,\,\,\,\,\,\,\,\,\,\,\,\,\,\,\,\,\,\,\,\,\,\,\,\,\,\,\,\,\,\,\,\,\,\,\,\,\,\,\,\,\,otherwise, \\
\end{aligned}
\right.
\end{equation*}    
where $k_b=20.0$. The $5^{th}$ car has the following acceleration:
\begin{equation*}
    a_{t,5}\!=\!\left\{
\begin{aligned}
&\!k_v(v_{s}\!-\!v_{t,5})\!-\!k_b(p_{t,3}\!-\!p_{t,5})\,\,if\,|p_{t,3}-p_{t,5}| \!<\! 13.0\\
&\!k_v(v_{s}\!-\!v_{t,5})\,\,\,\,\,\,\,\,\,\,\,\,\,\,\,\,\,\,\,\,\,\,\,\,\,\,\,\,\,\,\,\,\,\,\,\,\,\,\,\,\,\,\,\,\,otherwise. \\
\end{aligned}
\right.
\end{equation*}
The model of the $4^{th}$ car is as follows: 

        \begin{equation*}
         x_{t+1,4}=\left[\begin{array}{cc}
                 1  & 0 \\
                 0 & 0 
        \end{array}\right ] x_{t,4}
         +
         \left[ \begin{array}{c}
         1 \\
         \frac{1}{\Delta T}
         \end{array}
         \right ] u_{t} \Delta T,
         \end{equation*} 
where $u_{t}$ is the velocity of the $4^{th}$ car, and also the control signal generated by the controller.

The reward signal is defined to minimize the difference between $u_t$ and $v_{s}$, and an additional reward of 1.5 is given at the timesteps when $d_t = p_{t,3}-p_{t,4}$, which is the distance between the $3^{rd}$ and $4^{th}$ car, is within $[9.0,10.0]$ defined as the region where $d_t$ is expected to be. 
The cost signal is $\left\lVert d_{t+1} - d_{\text{desired}}\right\rVert $, where $d_{\text{desired}}=9.5$. 
CBFs are defined as $h_1(x_t) =p_{t,3}-p_{t,4} - \delta$ and $h_2(x_t) =p_{t,4}-p_{t,5} - \delta$, where $\delta$ is the minimum required distance between the cars. When a safety constraint is violated if two types of constraints cannot be simultaneously satisfied, the $4^{th}$ car may be too close to the $5^{th}$ car to make $d_t$ be within $[9.0,10.0]$. In this case, the backup controller where $u_{\text{nominal}}$ is simply set to be zero to check whether safety can be achieved as quickly as possible and then maintained is activated. The RL-based controller will be reinstated once the $4^{th}$ car leaves the dangerous area, namely beyond the proximity of the $5^{th}$ car.  
\begin{figure*}[ht]
        \centering
        \includegraphics[scale=0.403]{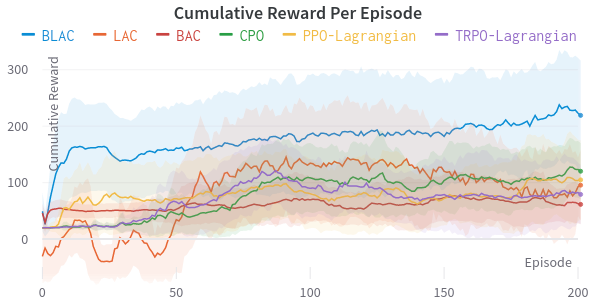}
        \hspace{5mm}
        \includegraphics[scale=0.403]{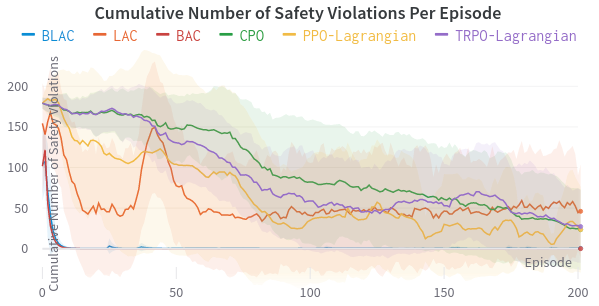}
        \caption{Comparisons are made between the cumulative reward and cumulative number of safety violations of each episode for the proposed BLAC (in blue) and other baselines in the simulated car following environment. Each plot shows the mean of ten experiments using different seeds, with the shading representing the standard deviation. Graph comparing the cumulative number of safety violations shows that the curves for BLAC and BAC decrease rapidly to zero and do not exhibit significant fluctuations afterwards.}
        \label{fig:Simulated Car Following results}
\vspace{-0.7cm} 
        
      \end{figure*}

Based on the results presented in Figure \ref{fig:Simulated Car Following results}, the proposed framework consistently yields the highest cumulative reward, indicating its superior performance in regulating the distance $d_t$ within the range $[9.0,10.0]$. Moreover, the much smaller cumulative number of safety violations compared to baselines demonstrates the effectiveness of the proposed framework in helping the system achieve safety with the assistance of CBFs. These results suggest that the proposed framework has promising potential for practical applications.

\section{CONCLUSIONS}
\label{sec:conclusion}
In this paper, we propose the BLAC framework, which combines separate CBF and CLF constraints with the actor-critic RL method to help to guarantee both the safety and stability of the controlled system. This framework imposes safety constraints for each step in the trajectory instead of the trajectory expectation which is widely used in previous research. Thus, this framework imposes stricter safety constraints, which is crucial in real-world safety-critical applications. Moreover, our framework contributes to guaranteeing the stability of the system, facilitating the system to approach the desired state (equilibrium) and obtain higher cumulative reward in tasks where high rewards are offered when the system gets closer to or reaches the desired state, such as navigation tasks. With the augmented Lagrangian method and backup controller, higher cumulative reward and fewer safety constraint violations are realized in the experiments.

However, there are also some limitations of this framework: 1. the CBFs are predefined before the learning process, but in real-world applications, it may be nontrivial to construct valid CBFs; 2. the framework requires knowledge of the control affine model of the system. Also,  performance comparisons can be conducted between this RL-based control policy and other model-based optimal control policies \cite{moyalan2021sum}; 3. the framework is only tested on tasks where the relative-degree of the CBFs is 1, and therefore, more research where CBFs with high relative-degree are used should be conducted in the future.  We believe that addressing these current limitations could be interesting future directions.






\bibliography{bibli_new}
\bibliographystyle{ieeetr}

\end{document}